\documentclass[journal,twocolumn]{IEEEtran}

\usepackage{enumerate}

\usepackage{amsmath,amsthm}
\usepackage[noend]{algpseudocode}
\usepackage{float}
\usepackage{color}
\usepackage{makeidx}
\usepackage{bbm}
\usepackage{graphicx}
\usepackage{lipsum}
\usepackage{soul}
\usepackage{tabularx}
\usepackage{dsfont}

\usepackage{amsfonts}
\usepackage{times}
\usepackage{graphicx}
\usepackage{latexsym}
\usepackage{dsfont}
\usepackage{amssymb}
\usepackage{cite}
\usepackage{verbatim}
\usepackage{subfigure}

\newtheorem{theorem}{Theorem}

\newtheorem{corollary}[theorem]{Corollary}

\newtheorem{proposition}{Proposition}

\newcommand{\figref}[1]{{Fig.}~\ref{#1}}

\def\bb0{{\mathbb{0}}}

\def\bb{{\mathbf{b}}}

\def\bff{{\mathbf{f}}}

\def\bh{{\mathbf{h}}}

\def\bx{{\mathbf{x}}}
\def\by{{\mathbf{y}}}

\def\b0{{\mathbf{0}}}

\def\bI{{\mathbf{I}}}

\def\bN{{\mathbf{N}}}

\def\bY{{\mathbf{Y}}}

\def\bbE{{\mathbb{E}}}

\def\sf0{{\mathsf{0}}}

\usepackage{epstopdf}
\newcommand{\sref}[1]{{Section}~\ref{#1}}

\def\rm{\mathrm}

\begin{document}
\title{Deep Learning for Massive MIMO with 1-Bit ADCs: When More Antennas Need Fewer Pilots}

\author{Yu Zhang, Muhammad Alrabeiah, and Ahmed Alkhateeb\\  {Arizona State University, Email: $\{$y.zhang, malrabei, alkhateeb$\}$@asu.edu} \thanks{This work is supported by the National Science Foundation under Grant No. 1923676.}}
\maketitle

\begin{abstract}

This paper considers uplink massive MIMO systems with 1-bit analog-to-digital converters (ADCs) and develops a deep-learning based channel estimation framework. In this framework, the prior channel estimation observations and deep neural networks are leveraged to learn the non-trivial mapping from quantized received measurements to channels. For that, we derive the sufficient length and structure of the pilot sequence to guarantee the existence of this mapping function. This leads to the interesting, and \textit{counter-intuitive}, observation that when more base-station antennas are employed, our proposed deep learning approach achieves better channel estimation performance, for the same pilot sequence length. Equivalently, for the same channel estimation performance, this means that when more antennas are employed, fewer pilots are required. This observation is also analytically proved for some special channel models. Simulation results confirm our observations and show that more antennas lead to better channel estimation in terms of the normalized mean squared error and the receive signal-to-noise ratio per antenna.
\end{abstract}

\section{Introduction} \label{sec:Intro}

% Motivation
Using low-resolution analog-to-digital converters (ADCs) in massive MIMO systems has the potential of reducing the power consumption while maintaining good achievable rate performance. These gains attracted increasing research interest in the last few years \cite{Choi2016, Mo2018}. One main challenge in these systems, however, is estimating the channels from highly quantized measurements. While several channel estimation techniques have been proposed in the literature \cite{Choi2016, Mo2018}, these solutions generally require very large training overhead (long pilot sequences). This highly impacts the feasibility of employing low-resolution ADCs in practical millimeter wave (mmWave) and massive MIMO systems. This paper targets overcoming these challenges and enabling low-resolution ADCs in  large-scale MIMO systems.

\textbf{Contribution:} In this paper, we propose a deep-learning based framework for the channel estimation problem in massive MIMO systems with 1-bit ADCs. In this framework, the prior channel estimation observations and deep neural networks are exploited to learn the mapping from the received quantized measurements to the channels. Learning this mapping, however, requires its existence in the first place. For that, we derive the sufficient length and structure of the pilot sequence that guarantees the existence of this quantized measurement to channel mapping.  Then, we make the interesting observation that for the same set of candidate user locations, more  antennas require fewer pilots to guarantee the mapping existence. This means that increasing the number of base station antennas reduces the number of required pilots to estimate their channels, which may seem counter-intuitive. The intuition justifying this observation, however, is that with more antennas, the quantized measurement vectors become more unique for the different channels. Hence, they can be efficiently mapped to their corresponding channels with less error probability. This observation is also proved analytically for the case of single-path channels. Simulation results highlight the promising gains of the proposed deep learning approach; they confirm that more antennas lead to better channel estimates, in terms of normalized mean-squared error (NMSE) and per-antenna signal-to-noise ratio (SNR).

\textbf{Prior Work:} To the best of our knowledge, no prior work has addressed deep learning based channel estimation for massive MIMO systems with 1-bit ADCs, or revealed the interesting observation that more antennas need fewer pilots. Relevant problems, however, have attracted a lot of research interest in the last few years \cite{Choi2016, Mo2018,Jeon2018,Balevi2019,Klautau2018,Gao2019}. In \cite{Choi2016, Mo2018}, classical signal processing tools have been leveraged to design low-complexity near-maximum likelihood (ML) data detector \cite{Choi2016} and develop efficient channel estimation techniques for massive MIMO and mmWave systems \cite{Mo2018}. A common problem in these solutions is the need for very long pilot sequences to achieve good data detection or channel estimation performance.

In \cite{Jeon2018,Balevi2019,Klautau2018,Gao2019}, machine learning techniques were developed to address several problems with 1-bit ADCs. For example, \cite{Jeon2018} developed a machine learning framework for data detection, but not for channel estimation, and \cite{Balevi2019} considered the channel estimation problem in OFDM systems, but only for systems with single antennas. In \cite{Klautau2018}, deep learning solutions for data detection and channel equalization were developed for MIMO systems, but their approach works only for low-dimensional MIMO regimes as the proposed network does not converge for large numbers of antennas. In \cite{Gao2019}, systems with mixed full-resolution and low-resolution ADCs were adopted, but only the received signals from full-resolution ADCs were used to estimate the channels, which is in fact equivalent to the channel mapping concept we proposed in \cite{Alrabeiah2019}.

%notation
%\textbf{Notation}: $\bA$ is a matrix, $\ba$ is a vector, $a$ is a scalar. $[\ba]_m$ is the $m$th element of $\ba$. $\bA^T$, $\bA^*$ and $\bA^H$ are the transpose, conjugate, and hermitian of $\bA$. $\angle a$ denotes the argument of the complex number $a$.

\begin{figure}[hbt]
	\centering
	\includegraphics[width=0.9\columnwidth]{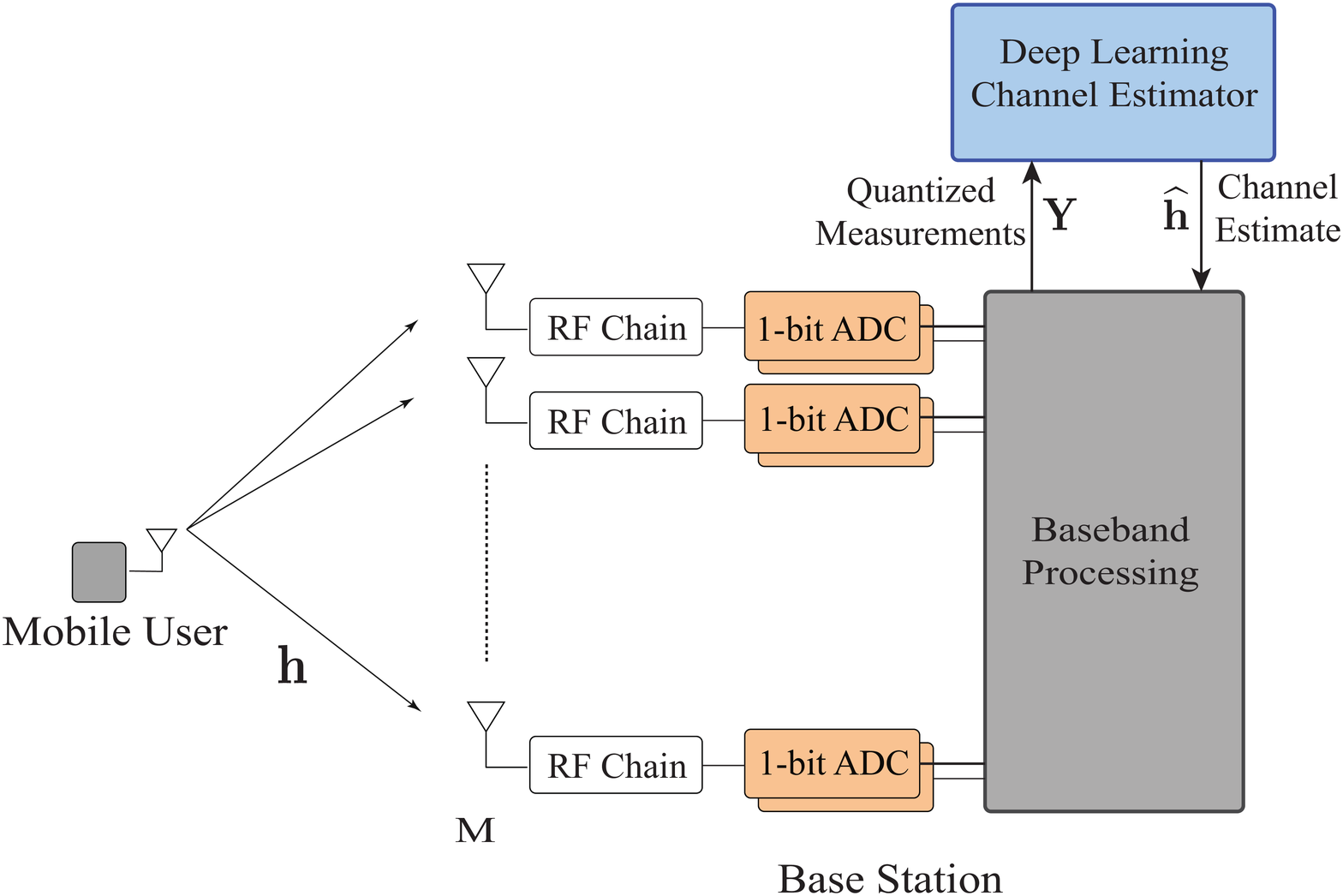}
	\caption{The adopted massive MIMO system where the base station receiver uses 1-bit ADCs. The uplink quantized received measurement matrix  $\bY$ is fed to a deep learning model that predicts the channel vector $\widehat{\bh}$.}
	\label{fig:System}
\end{figure}

\section{System and Channel Models} \label{sec:System}

We consider the system shown in \figref{fig:System} where a massive MIMO base station (BS) with $M$ antennas is communicating with a single-antenna user. The BS employs only 1-bit analog-to-digital converters (ADCs) in its receive chains. Further, we adopt a time division duplexing (TDD) system operation where the channel is estimated through an uplink training and used for downlink data transmission as summarized below.

\textbf{Uplink Training:}
If the user transmits an uplink pilot sequence $\bx\in \mathbb{C}^{N\times 1}$, where $N$ is the length of pilot sequence, then the received signal at the BS after the ADC quantization can be expressed as
\begin{equation}\label{recSig}
\bY = \rm{sgn}(\bh \bx^{T} + \bN),
\end{equation}
where $\bh \in \mathbb{C}^{M \times 1}$ is the channel vector between the mobile user and the BS antennas, $\bN$ is the receive noise matrix at the BS with independent and identically distributed (i.i.d.) elements drawn from $\mathcal{N}_\mathbb{C}(0, \sigma^2)$, and the transmitted pilot sequence is satisfying $\bbE\left[\bx \bx^H\right]=P_t \bI$ with $P_t$ denoting the average transmit power per symbol. The element-wise operator $\rm{sgn}(\cdot)$ is the signum function, and it is applied separately to the real and imaginary part of its argument. Finally,  ${\bf Y}$ is the $M \times N$ quantized receive measurement matrix that consists of the received pilot signals after quantization.

%(outputs $+1$ if the input argument is bigger than zero or $-1$ otherwise),

\textbf{Channel Model:} We adopt a general geometric channel model for $\bh$. Assume that the signal propagation between the user and BS consists of $L$ paths. Each path $\ell$ has a complex gain $\alpha_\ell$ and an angle of arrival $\phi_\ell$, then we define
\begin{equation}\label{channel}
{\bf h} = \sum\limits_{\ell=1}^L\alpha_\ell{\bf a}(\phi_\ell),
\end{equation}
where ${\bf a}(\phi_\ell)$ is the array response vector of the BS.

\textbf{Channel Estimation:} The quantized receive measurement matrix $\bY$ will be processed using a channel estimator to construct an estimated channel vector $\widehat{\bh}$. This paper focuses on leveraging deep learning models for this channel estimation task as will be discussed shortly in Sections \ref{sec:Problem}, \ref{sec:DL_solution}.

\textbf{Downlink data transmission:} Based on the estimated channel vector, the downlink beamforming $\bff$ is constructed as a conjugate beamforming, i.e., $\bff=\textstyle{{\widehat{\bh}}^{^{^*}}}/\|{\widehat{\bh}}\|$.  With this design, the downlink receive SNR per \textit{transmit} antenna can be written as
\begin{equation} \label{eq:per_ant}
\mathsf{SNR}_\text{ant}=\frac{\rho}{M} \frac{\left|\widehat{\bh}^H \bh\right|^2}{\|\widehat{\bh}\|^2},
\end{equation}
where $\rho$ denotes the average receive SNR before beamforming.

\section{Problem Definition} \label{sec:Problem}

In this paper, we investigate the design of efficient channel estimation strategy that can construct the channel $\widehat{\bh}$ from the highly quantized received signal $\bY$. More specifically, assuming that the pilot sequence $\bx$ is known to the BS (the receiver), our objective is to develop a channel estimation strategy that minimizes the normalized mean-squared error (NMSE) between the estimated channel  and the original channel vectors, defined as
\begin{equation}\label{NMSE}
  \mathsf{NMSE} = \mathbb{E}\left[\frac{\|{\bf h}-{\widehat{{\bf h}}}\|^2}{\|{\bf h}\|^2}\right].
\end{equation}

Traditionally,  channel estimation algorithms attempt to process the  quantized signal ${\bY}$ to estimate the channel ${\bh}$. Since ${\bY}$ is highly quantized, though, very long pilot sequences normally need to be used to achieve a reasonable channel estimation quality. To overcome this challenge, we propose to exploit deep learning models to learn how to efficiently estimate the channels from the quantized measurements while requiring only small pilot sequences.

\section{Deep Learning Based Channel Estimation} \label{sec:DL_solution}

Classical channel estimation techniques for massive MIMO systems with low resolution ADCs, such as \cite{Choi2016, Mo2018}, attempt to estimate the channel only from the quantized received signal, without using prior observations. The channels, however, are intuitively some functions of the different elements of the environment, such as the environment geometry, materials, transmitter/receiver positions, etc. \cite{Alkhateeb2018}. This means that the BSs deployed in a certain environment will likely experience similar channels more than once. Hence, prior experience could be leveraged to learn the underlying relation between the quantized received signals and the channels. This has the potential of significantly reducing the pilot length. With that in mind, we propose to utilize deep learning to learn the mapping from the quantized received measurement matrix ${\bf Y}$ to the channel ${\bf h}$. Next, we first establish the conditions under which this mapping exists then highlight an interesting observation about how scaling the number of antennas up scales the number of required pilots down.

%This strategy has the potential of significantly reducing the pilot length as will be discussed in this section. With this motivation,

\subsection{Mapping Quantized Measurements to Channels} \label{subsec:map}

Now, we investigate the existence of the mapping from quantized measurements to channels and highlight the motivation to leverage deep learning models. We also briefly explain why our proposed approach has the potential of reducing the number of pilots. First, consider a setup (indoor or outdoor environment) where a massive MIMO BS is serving a single-antenna user, as described in \sref{sec:System}. Let $\{\bh\}$ denote the set of candidate channels for the user, which depends on the candidate user positions as well as the surrounding environment. Further, let $\{\bY\}$  represent the corresponding quantized measurement matrices for the channel set $\{\bh\}$ and a pilot sequence $\bx$. We can then define the mapping from quantized measurements to channels, $\boldsymbol{\Phi(.)}$, as
\begin{equation}
\boldsymbol{\Phi}: \{\bY\} \rightarrow \{\bh\}.
\end{equation}

\noindent Note that if this mapping exists and is known, it can be used to predict the channel vector $\bh$ from the quantized receive matrix $\bY$. Next, we establish existence of this mapping in Proposition\ref{prop1} and then discuss how we can learn it.

\begin{proposition} \label{prop1}
	Consider the system and channel model in \sref{sec:System} with $\bN=0$ and a set of candidate channels $\{\bh\}$. Define the angle $\alpha$ as	
	\begin{equation}\label{minAngle}
	{\alpha} = \min\limits_{\substack{\forall \bh_u, \bh_v \in \{\bh\}\\ u \neq v}} \max\limits_{\forall m}\left|\angle \left[\bh_u\right]_{m} - \angle \left[\bh_v\right]_{m}\right|.
	\end{equation}
	If the pilot sequence $\bx$ is constructed to have a length $N$ satisfying  $N \geq \left\lceil\frac{\pi}{2 \alpha} \right\rceil$, with the angles of the pilot complex symbols uniformly sampling the range $]0, \frac{\pi}{2}]$, then the mapping function $\boldsymbol{\Phi}(.)$ exists.
\end{proposition}
\begin{proof}
	See Appendix \ref{app1}.
\end{proof}
Proposition \ref{prop1} means that once the pilot sequence is designed with the specific structure described in the proposition, then there exists a one-to-one mapping $\boldsymbol{\Phi}(.)$ that can map the quantized measurement matrix $\bY$ to the channel $\bh$, i.e., it can use $\bY$ to predict $\bh$. Interestingly, as we will see in \sref{sec:Results}, only a few pilot symbols (very small N) are needed in massive MIMO systems to make this mapping $\boldsymbol{\Phi}(.)$ exist with high probability. This has the potential of significantly reducing the channel training overhead compared to classical 1-bit ADC channel estimation techniques. To be able to leverage this mapping function, however, we need to know it. Characterizing this mapping analytically is very non-trivial mainly because of the non-linear quantization. \textbf{With this motivation, we propose to exploit the powerful learning capabilities of deep neural networks to learn this mapping and harvest the promising gains in reducing the channel training overhead}. Before describing the adopted deep learning model in \sref{sec:DL_model}, we first highlight in the next subsection an interesting gain of using the proposed deep learning based 1-bit ADC channel estimation approach in massive MIMO systems.

\subsection{Discussion: More Antennas Need Fewer Pilots} \label{subsec:more}
As shown in Proposition \ref{prop1} and its proof, the desired pilot sequence should have a length that guarantees that every two different channels in $\{\bh\}$ result in two unique quantized measurement matrices. Intuitively for the same uplink pilot sequence length, the more antennas deployed at the BS the more likely they lead to unique measurement matrices. Interestingly, this means that more antennas will lead to better channel estimation as will be shown in \sref{sec:Results}. This also implies that when more antennas are employed at the BS, fewer pilots are required to guarantee that the mapping from $\{\bh\}$ to $\{\bY\}$ is bijective (one-to-one). This interesting relation can be analytically characterized for several channel models. In the next corollary, we consider the LOS channel model, and prove that more antennas need fewer pilots.
\begin{corollary} \label{cor1}
	Consider a BS employing a ULA with half-wavelength antenna spacing and a single-path channel model ($L=1$). Let $\delta \phi$ denote the smallest difference between any two angles of arrival, $\phi_1, \phi_2 \in [0, \pi[$, of any two users. If the pilot sequence is constructed according to Proposition \ref{prop1}, then the required pilot sequence length to guarantee that the mapping $\boldsymbol{\Phi}(.)$ exists is
	\begin{equation}
	N=\left\lceil{\frac{1}{(M-1)(4 \sin^2 (\delta \phi/2) )}}\right\rceil.
	\end{equation}
\end{corollary}
\begin{proof}
	The proof follows from Proposition \ref{prop1} and is omitted because of the limited space.
\end{proof}
Corollary \ref{cor1} captures clearly the interesting gain in our proposed deep learning approach, where more antennas require fewer pilots to ensure the existence of $\boldsymbol{\Phi}(.)$, i.e., to have the same channel estimation quality. This interesting finding will also be validated using numerical simulations in section \sref{sec:Results}.

\subsection{Proposed Deep Learning Model} \label{sec:DL_model}

To map back quantized received signal to complex-valued channels, we chose to utilize the expressive ability of deep learning \cite{DLBook}, more specifically fully-connected neural networks. These networks are known to be good function approximators \cite{NNUnivApprox}, and therefore, we design and train a dense neural network to learn the mapping from quantized measurements to channels.

\textbf{Network architecture:} The designed network has three dense stacks of layers. The first two are very wide and comprise a sequence of fully-connected layer, non-linearity layer, and dropout layer. The number of neurons in each fully-connected layer is $L_\text{NN}$, and they are followed by Rectified Linear Unit (ReLU) non-linearities. The last stack, which is the output, has only a fully-connected layer with $2M$-neurons.

\textbf{Network training}: As our target is estimating users' channels, we pose our learning problem as a regression problem in a supervised leaning setting; the network is trained to minimize a loss function measuring the accuracy of the predictions using a distance measure to some \textit{desired} outputs. Given the channel estimation problem formulation in \sref{sec:Problem}, we choose NMSE as the loss function. Training is implemented using ADAM optimizer, and, therefore, we choose to minimize the average NMSE over a training mini-batch of size $B$ users.

\textbf{Data pre-processing}: Before any training takes place, the inputs and outputs of the network must be pre-processed for efficient training, see \cite{EffBackProp}. The first stage of pre-processing normalizes those channels, whether in the training or testing datasets, to the range $[-1,1]$ using the maximum absolute channel value obtained from the training set. We have found such normalization to be very useful in prior work \cite{Alrabeiah2019,Alkhateeb2018}. The second stage of the pre-processing is vectorizing the received quantized measurement matrices to have dimensions of $M N \times 1$. Finally, since popular deep learning software frameworks mainly support real-valued computations, channel and measurement vectors are decomposed into real and imaginary components and flattened into a ($2 M\times 1$) vectors for the channels and  $2 M N$-dimensional vectors for the measurement.

% Formally:
%\begin{equation}
%   [\bar{\bh}_u]_m =  [{\bh}_u]_m/\Delta_{max},\ \ \forall m \in {1, 2, ..., M}
%\end{equation}
%and:
%\begin{equation}
%  \Delta_{max} = \underset{\forall u, m} \max \left|[\bh_u]_m\right|.
%\end{equation}

\section{Simulation Results} \label{sec:Results}

In this section, we evaluate the performance of our proposed deep learning based channel estimation solution for massive MIMO systems with 1-bit ADCs. We first describe the adopted scenario, dataset and the deep learning parameters used in our simulations and then discuss the interesting results.

\begin{figure}[t]
	\centering
	\includegraphics[width=.99\columnwidth]{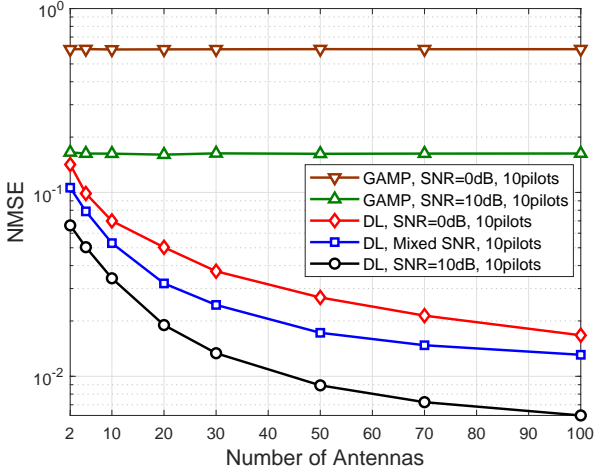}
	\caption{The NMSE of the predicted channel using the proposed deep learning based approach and EM-GM-GAMP versus the number of BS antennas, for different pilot sizes and SNRs.}
	\label{fig:SimuPilot}
\end{figure}

\subsection{Scenario and Dataset} \label{subsec:dataset}

In our simulations, we consider the indoor  massive MIMO scenario `I1$\_$2p4' which is offered by the DeepMIMO dataset \cite{DeepMIMO} and is generated based on the accurate 3D ray-tracing simulator Wireless InSite \cite{Remcom}. This scenario depicts a 10m$\times$10m room with two tables and users distributed across two x-y grids.

%\begin{table}[h!]
%  \begin{center}
%  \caption{The adopted DeepMIMO dataset parameters}
%  \label{table1}
%  \resizebox{0.9\columnwidth}{!}{
%    \begin{tabular}{|c|c|}
%     \hline
%     % after \\: \hline or \cline{col1-col2} \cline{col3-col4} ...
%     Parameter & value \\
%     \hline
%     Scenario name & I1$\_$2p4 \\
%     \hline
%     Active BSs & 32 \\
%     \hline
%     Active users & Row 1 to 502 \\
%     \hline
%     Number of BS antennas in (x, y, z) & (1, 100, 1) \\
%     \hline
%     System bandwidth & 0.01 GHz \\
%     \hline
%     Number of OFDM sub-carriers & 1 (single-carrier) \\
%     \hline
%     OFDM sampling factor & 1 \\
%     \hline
%     OFDM limit & 1 \\
%     \hline
%     Number of multipaths & 10 in \figref{fig:SimuPilot} and 1 in \figref{fig:SimuSNRperAnt}\\
%     \hline
%    \end{tabular}
%  }
%  \end{center}
%\end{table}

Given this ray-tracing scenario, we generate the DeepMIMO dataset which contains the channels between every candidate user location and every antenna at the BS. We adopt the following DeepMIMO parameters: (1) Scenario name: I1$\_$2p4, (2) Active BSs:  32, (3) Active users: Row 1 to 502, (4) Number of BS antennas in (x, y, z): (1, 100, 1), (5) System bandwidth: 0.01 GHz, (6) Number of OFDM sub-carriers: 1 (single-carrier), (7) Number of multipaths: 10 in \figref{fig:SimuPilot} and 1 in \figref{fig:SimuSNRperAnt}. To form training and testing datasets, we first shuffle the elements of the generated DeepMIMO dataset, and split it into 70\% training set and 30\% testing set. These datasets are then used to train the deep learning model and evaluate the performance of the proposed solution.

%with 70\%  of the total size and a testing dataset with the other 30\%

\subsection{Model Training and Testing} \label{subsec:Model}

The fully-connected network adopted in our simulation has two hidden layers, with 8192 neurons each\footnote{The code files that implement our deep learning based channel estimation solution are available in \cite{code_1bit}.}. The dimension of input and output layers depends on the number of antennas at BS and the number of pilots used for estimation. For example, if there are 100 antennas and 10 pilot symbols, then the input and output sizes will be 2000 and 200 respectively. We organize our training samples by $({\by}, {\bh})$, with ${\bf y}$ stands for the input and ${\bf h}$ for the target channel. Each sample corresponds to one user and all those samples are randomly drawn from the two grids.

The network is trained with 105981 samples for 100 epochs, and the performance of the trained network is evaluated with 45421 unseen samples. We use NMSE as our performance metric. \textbf{To guarantee a fair comparison, the structure of the network and training parameters are kept unchanged in all our simulations}, except for the input and output dimensions. 
%We implement all the deep learning modeling, training and testing  using the MATLAB Deep Learning Toolbox.

\begin{figure}[t]
	\centering
	\includegraphics[width=.99\columnwidth]{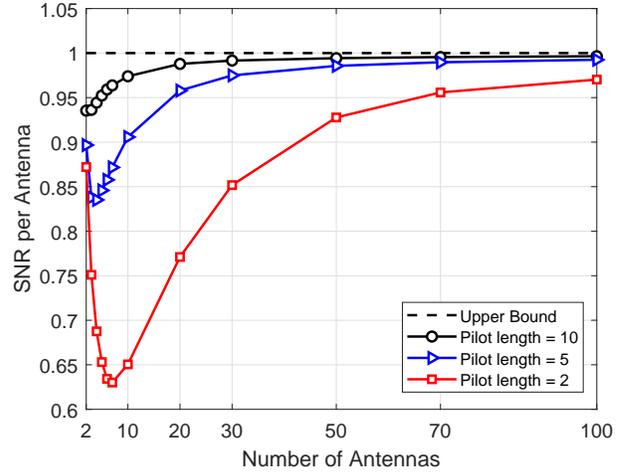}
	\caption{Achievable SNR per antenna for the proposed deep learning based channel estimation approach and the upper bound for different BS antenna numbers and pilot sizes.}
	\label{fig:SimuSNRperAnt}
\end{figure}

\subsection{Performance Evaluation}

Now, we evaluate the performance of our proposed deep learning based channel estimation approach, adopting the described uplink massive MIMO scenario in Sections \ref{subsec:dataset} and \ref{subsec:Model}.
In \figref{fig:SimuPilot}, we plot the NMSE versus the number of antennas ($M$) for different SNRs (0dB, 10dB, and mixed SNRs in the range of 0dB-10dB). This is implemented by adding noise samples to the measurement matrices used to train and test the deep neural network.
\figref{fig:SimuPilot} first shows that the proposed deep learning approach requires only a few pilots to have very accurate channel predictions at different SNRs. This is in contrast with the very long sequences normally required by classical (non-machine learning) channel estimation approaches \cite{Choi2016, Mo2018}, such as expectation maximization Gaussian-mixture generalized approximate message passing (EM-GM-GAMP) \cite{Mo2018}. As shown in \figref{fig:SimuPilot}, with the very short pilot sequences, the proposed deep learning approach has a clear gain over the EM-GM-GAMP solution..

Furthermore, \textbf{\figref{fig:SimuPilot} illustrates that the NMSE performance of the proposed solution improves significantly when more antennas are deployed at the BS, which confirms our results in \sref{subsec:more}.}  
This interesting result is consistent with Proposition I. The minimum $\alpha$ (in rad) defined by (6) of the adopted dataset is $3.07\times10^{-5}$ for the system with 2 antennas and $0.2476$ for the 100 antenna system. Accordingly, based on Proposition I, the minimum pilot lengths guaranteeing the uniqueness of all the channels (in the form of quantized measurements) are $51166$ and $7$, respectively. This clearly shows the potential of our deep learning based approach that requires very small pilot sequences for massive MIMO systems. Further, even though Proposition 1 suggests that the number of pilots required for full bijectiveness (i.e., to be able to distinguish between any two channels) is very large for the case of 2 antennas, we can show that the \textit{percentage} of these channels that need long pilots to be distinguishable is negligible.  For example, with only 5 pilots,  $98\%$ of the \textit{complex} channels in the dataset are distinguishable. This percentage increases to $99.5\%$ with 10 pilots. All that clearly explains the good performance achieved by the proposed solution with only few pilots.

In \figref{fig:SimuSNRperAnt}, we repeat the same experiment of \figref{fig:SimuPilot} and compare the SNR \textit{per antenna} defined in \eqref{eq:per_ant}, for different pilot length and antenna numbers. In the figure, the SNR of the received measurement matrices is fixed at $0$dB. \figref{fig:SimuSNRperAnt} shows that for a fixed pilot length, despite a dip in the region of small number of antennas, the  achievable per-antenna SNR approaches the upper bound as the antenna number increases. This upper bound is achieved by designing the conjugate beamformer based on the exact channel knowledge. Interestingly, this is the case even when only 2 pilots are used, i.e., with $N=2$. The reason that we get the dip (in case of $N=2, 5$) is mainly because the improvement on the total SNR does not match the rate at which antenna number increases. However, we notice that when increasing the pilot length or antenna numbers, this dip gradually vanishes. This can be explained by the insights in Sections \ref{subsec:map} and \ref{subsec:more}, which conclude that the mapping from the quantized measurements to channels becomes more bijective as more pilots are sent or more antennas are employed.

\section{Conclusion}
In this paper, we developed a deep learning based channel estimation framework for massive MIMO systems with 1-bit ADCs. We derived the structure and length of the pilot sequence that guarantee the existence of the mapping from quantized measurements to channels. We then showed that this existence requires fewer pilots for larger antenna numbers. This was confirmed using both analytical and simulation results which showed that only few pilots are required to efficiently estimate massive MIMO channels. Further, the results showed that the achievable SNR per antenna approach the upper bound with increasing the number of antennas, which highlights the promising gains for massive MIMO systems. For  future work, it is interesting to extend the proposed approach for broadband systems with frequency-selective channels. 
\begin{appendices}
	\section{} \label{app1}

\textit{Proof of Proposition} \ref{prop1}: For the mapping $\boldsymbol{\Phi}: \{\bY\} \rightarrow \{\bh\}$ to exist, the (inverse) mapping $\boldsymbol{\Psi}:  \{\bh\} \rightarrow \{\bY\}$ has to be bijective. To guarantee the bijectiveness of the mapping $\boldsymbol{\Psi}$, any two channels in $\{\bh\}$ have to lead to two unique received quantized measurements in $\{\bY\}$. Consider any two channels $\bh_u, \bh_v \in \{\bh\}$, and a pilot sequence $\bx$, the following condition needs then to be satisfied
\begin{equation}\label{unique_map}
\rm{sgn}\left( {\bh_u}{\bx}^T  \right) \neq \rm{sgn}\left(  {\bh_v} \bx^T  \right).
\end{equation}
This can be achieved if we guarantee that there exists at least one pair of corresponding elements in \eqref{unique_map} that satisfies
\begin{equation}\label{eq:notequal}
\rm{sgn}\left(\left[\bh_u\right]_m \left[\bx\right]_n\right) \neq \rm{sgn}\left(\left[\bh_v\right]_m \left[\bx\right]_n\right),
\end{equation}
for any $m \in {1, ..., M}, n \in {1, ..., N}$.  Since the elements of the channel and pilot vectors are complex values, we can view these elements as vectors in the complex plane. If we write $\left[ \bx \right]_n$ as $\left|x_n \right| e^{j \theta_n}$, then multiplying the channel element $\left[\bh_u\right]_m$ by $\left[\bx\right]_n$ simply rotates $\left[\bh_u\right]_m$ by $\theta_n$ and scales it by $|x_n|$. In this sense, achieving \eqref{eq:notequal} is by ensuring that  $\left[\bh_u\right]_m $ and $\left[\bh_v\right]_m$ are rotated by $\left[\bx\right]_n$ to lie in two different  quadrants of the complex plane. This can happen if we have  a pilot sequence of length $N \geq \left\lceil \frac{\pi}{2 \overline{\alpha}} \right\rceil$, where
\begin{equation}\label{angle}
\overline{\alpha} = \max\limits_{\forall m}\left|\angle \left[\bh_u\right]_m - \angle\left[\bh_v\right]_m\right|,
\end{equation}
and draw these $N$ pilots such that their angles evenly sample the range $]0, \frac{\pi}{2}]$. Finally, to ensure that any two channels in $\{\bh\}$ satisfy the condition in \eqref{unique_map}, we find the minimum $\overline{\alpha}$ defined in \eqref{angle} among all channel pairs in $\{{\bf h}\}$, that is
\begin{equation}\label{minAngle2}
{\alpha} = \min\limits_{\substack{\forall \bh_u, \bh_v \in \{\bh\}\\ u \neq v}} \max\limits_{\forall m}\left|\angle \left[\bh_u\right]_{m} - \angle \left[\bh_v\right]_{m}\right|.
\end{equation}
This leads to $N \geq \left\lceil \frac{\pi}{{2 \alpha}} \right\rceil$, which concludes the proof.

\end{appendices}

%------------------------------------------------------------------------------------------------------------
%\linespread{1.3}
% Generated by IEEEtran.bst, version: 1.14 (2015/08/26)

\end{document}